\documentclass[11pt]{article}
\usepackage[utf8]{inputenc}
\usepackage{vitercik}
\usepackage{float}
\usepackage{tikz}
\usepackage{color-edits}
\usepackage[normalem]{ulem}
\usepackage{chngpage}
\usepackage{tabularx}
\usepackage{subcaption}

\newcommand{\app}{\textnormal{app}}

\title{Sorting from Crowdsourced Comparisons using Expert Verifications}
\date{\today}
\author{
	Ellen Vitercik \\ Stanford \\ \texttt{vitercik@stanford.edu}
  \and
	Manolis Zampetakis \\ Yale \\ \texttt{emmanouil.zampetakis@yale.edu}
 \and
	David Zhang \\ UC Berkeley \\ \texttt{david.z@berkeley.edu}}
 
\begin{document}

\maketitle

\begin{abstract}
We introduce a novel noisy sorting model motivated by the Just Noticeable Difference (JND) model from experimental psychology. The goal of our model is to capture the low quality of the data that are collected from crowdsourcing environments. Compared to other celebrated models of noisy sorting~\citep[e.g.,][]{feige1994computing, Braverman16:Parallel, Gu23:Noisy}, our model does not rely on precise data-generation assumptions and captures crowdsourced tasks' varying levels of difficulty that can lead to different amounts of noise in the data. To handle this challenging task, we assume that we can verify some of the collected data using expert advice. This verification procedure is costly; hence, we aim to minimize the number of verifications we use. 

We propose a new efficient algorithm called CandidateSort, which we prove uses the optimal number of verifications in the noisy sorting models we consider. We characterize this optimal number of verifications by showing that it is linear in a parameter $k$, which intuitively measures the maximum number of comparisons that are wrong but not inconsistent in the crowdsourcing data.
\end{abstract}

\section{Introduction}

 Sorting from noisy data is a fundamental problem in computer science that arises in many practical applications, from social choice \cite{caplin1991aggregation} to crowdsourced feedback \cite{Capuano16:Fuzzy, Kazai11:Crowdsourcing, Lang11:Mechanical, Shah13:MOOC, Von08:Recaptcha, Wang22:MOOC} to peer grading in online courses \cite{piech2013tuned}. In the context of crowdsourced feedback or peer grading, studies show that using ordinal information is preferable to using cardinal information~\citep[e.g.,][]{Voorhees98:Variations}, i.e., it is preferable to use noisy comparisons as opposed to noisy evaluations. This effect is primarily due to the miscalibration arising from variance in participants' judgment.
One instantiation of this appears in machine translation, where the goal is to rank a list of candidate translations for a given sentence or phrase \cite{Li2020:Crowdsourced}.  Online users fluent in the given languages can provide the large-scale input required for translation systems to function smoothly.
 
Due to its relevance, sorting from noisy comparisons has become an important theoretical problem in computer science with many old and recent results on this topic, covering the complexity of finding a correct sorting from noisy comparisons \cite{feige1994computing, Wang22:Noisy, Gu23:Noisy}, the parallel complexity of this problem \cite{Braverman16:Parallel}, and faster algorithms with more structured noisy feedback \cite{Braverman08:NoResampling}.
 
  The classical formulation of this problem \cite{feige1994computing} assumes that every comparison is independently correct with some probability $p > 1/2$, which is the same across all comparisons. This simple model ignores the fact that, in many cases, some comparisons are easier than others. For example, comparing the first with the last candidate is often easier than comparing two candidates with almost the same ordering. The Mallows model is a more structured model that accounts for this asymmetry~\cite[e.g.,][]{Braverman08:NoResampling}. In this model, the probability of observing an ordering depends on its distance to the ground truth ordering. Although this is a natural model in many instances, it assumes a very precise distributional prior for the input data, which might not align with the unpredictable nature of data we observe in crowdsourcing or peer grading. 

\subsection{Our Contribution} \label{sec:contribution}
We introduce a new model of noisy comparisons with a \textit{verification oracle} and provide efficient algorithms for determining the true ordering of the candidates while using as few verifications as possible. Our model accounts for the discrepancy in comparing candidates with different distances in the true ranking while not making strong assumptions on the input data.

Our model is inspired by the concept of \textit{just-noticeable differences (JND)} in experimental psychology, as introduced by the celebrated psychologist Ernst Heinrich Weber \cite{Weber1831:Pulsu}. The key idea is that comparisons may be easy as the difference in quality between two candidates decreases, but only until a threshold where two candidates have near-indistinguishable quality. Comparisons then become arbitrary: the probability of fallacious comparisons can become even larger than $1/2$.
More formally, in this JND model, we assume all candidates have a latent cardinal value representing their quality. There is a threshold $\delta > 0$ (intuitively small) such that comparisons are correct when the difference in quality exceeds the threshold. Otherwise, the comparisons are arbitrary.  

To handle the existence of arbitrary comparisons, we assume that we have access to an \emph{expert} who can give correct pairwise comparisons regardless of the candidate pair, perhaps due to experience and training. In the rest of the paper, we call the expert-provided comparisons \textit{verifications} since they are guaranteed to be correct. Of course, the expert's input may be much more expensive than the crowdsourced data, and hence, we aim to minimize the amount of feedback we require from the experts. This technical problem is the focal point of our study.

We formalize and examine this problem by developing a framework where we first receive many noisy pairwise comparisons under a noise model subsuming the JND model. We then aim to determine a correct ranking from this input using as few verifications as possible. We now give more technical details on our model and results.

\subsubsection{Our Model} \label{sec:intro:model}

In order to capture the described setting, we introduce a model where a ranking of a set $V$ of $n$ elements is to be determined from two sources of comparisons:
\begin{description}
\item[Crowdsourcing data.] This data set includes the outcome of a set of pairwise comparisons that have been crowdsourced.  To represent these data, we define a directed multi-graph $G = (V, E)$. The direction of every edge $(x_i, x_j)$ in $G$ represents the outcome of comparing the elements $x_i$ and $x_j$. We allow $G$ to be a multi-graph because we may request the same comparison multiple times from different workers in the crowdsourcing platform. Because the data in the graph $G$ are noisy, $G$ may have cycles, although an underlying true ordering of the elements $V$ exists.

\item[Input from experts.]
Apart from the graph $G$, the algorithm can ask for the comparison between $x_i$ and $x_j$ from some expert. The answer from the expert is certainly correct, but we assume that every such query to an expert is costly. We call every such query a \emph{verification}, and one of our goals is to minimize the number of verifications.
\end{description}

The noisy comparisons may be thought of as crowdsourced data or data from cheap but potentially unreliable sources, such as machine learning models.
The verifications may be data provided by experts whose judgment regarding orderings is accurate, but their input is more expensive than that of the crowdsourcers.

We introduce three different models that capture the generation and quality of the crowdsourcing data. We characterize the number of verifications needed in each of these models. Our first model corresponds to experimental psychology's just noticeable difference (JND) model. The other two models are abstractions of the JND model, illustrating that our results are general and are not tailored to the specific details of the JND model.  

\begin{description}
  \item[Just Noticeable Difference Model.] In this model, we assume that every element $x_i$ of $V$ has some cardinal value $\ell_i$. The comparison between two elements $x_i$ and $x_j$ from crowdsourcing in this model has a uniformly random outcome if $|\ell_i - \ell_j| \le \delta$ and it is correct if $|\ell_i - \ell_j| > \delta$. In this model, we can pick a parameter $r$ representing the number of times we independently ask to compare each pair of elements.
  \item[Fixed Ambiguity Model.] In this model, we control the number of inconsistent comparisons that we receive from the crowdsourced data. A comparison between $x_i$ and $x_j$ in the crowdsourced data is inconsistent if the graph $G$ contains both of the directions: $(x_i, x_j)$ and $(x_j, x_i)$. In this model, we assume that every vertex $x_i$ of $G$ has the same number of inconsistencies with smaller elements---i.e., with elements $x_j$ smaller than $x_i$---and the same number of inconsistencies with larger elements. We precisely define this model in Definition \ref{def:ambiguous}.
  \item[Adversarial Model.] In this model, we assume that we start from a graph from the fixed ambiguity model, and we allow an adversary to corrupt $k$ inconsistent edges of $G$ and make them look like consistent edges. As we show in Section \ref{sec:results}, the JND model is a special case of this adversarial model with $O(n \delta 2^{-r})$ expected corruptions.
\end{description}

\subsubsection{Our Results} \label{sec:intro:results}

We propose a novel algorithm, \textsf{CandidateSort}, and we establish its efficiency in the three different models that we defined above:
\begin{enumerate}
    \item \textbf{Fixed Ambiguity Model.} In Theorem \ref{algo-correct}, we show that \textsf{CandidateSort} will efficiently find the correct ordering using \textbf{zero} verifications in this model.
    \item \textbf{Adversarial Model.} In Theorem \ref{thm:general-result}, we prove that \textsf{CandidateSort} requires only $O(k)$ verifications to find the correct ordering in the adversarial model with $k$ corruptions. This result is surprising since the algorithm does not know which of the $O(n^2)$ edges are corrupted.
    \item \textbf{Just Noticeable Difference Model.} In Theorem \ref{expected-verifications-discrete}, we prove that \textsf{CandidateSort} only needs $O(n \delta 2^{-r})$ verifications in the JND model to efficiently compute the correct ordering with high probability.
\end{enumerate}
In Section \ref{sec:lb}, we prove that the results from the adversarial model and the JND model are tight: no algorithm computes the correct ordering and with asymptotically fewer verifications.

Compared to other celebrated models of noisy sorting (see Section \ref{sec:related}), our model does not rely on precise data-generation assumptions and captures crowdsourced tasks' varying levels of difficulty that can lead to different amounts of noise in the data. We therefore believe that our \textsf{CandidateSort} algorithm can be a powerful technique for denoising data from crowdsourcing environments. We also believe that our model can inspire more theoretical results in noisy sorting without precise assumptions about the data generation procedure.

\subsection{Related Work} \label{sec:related}

Our model of noisy comparisons is based on the just noticeable difference (JND) model that is a product of the celebrated Weber's Law in psychology \cite{Warren1922:Elements}. The application of Weber's Law in the context of comparisons provided by human input has led to noisy sorting models that have been extensively used in many experimental psychology settings~\cite[e.g.,][]{David88:PairedComparisons, Thurstone27:Comparative}. 
 
\citet{Ajtai15:JND} and \citet{Acharya18:JND} have also explored the theoretical properties of this model. However, they aim to provide efficient algorithms that return the best possible sorting from the given input without using verifications. Beyond the JND model, there are other noisy sorting models under which the probability of an incorrect comparison between two alternatives depends on some notion of distance between the alternatives. The most popular are the Mallows model \cite{Mallows1957:Non,Braverman08:NoResampling,Busa-FeketeFSZ19:Optimal}, the Plackett-Luce model \cite{Plackett1975:Analysis,Luce12:Individual, MaystreG15:fast}, and the Bradley-Terry-Luce \cite{BradleyT1952:Rank, Hunter04:MM}. Pairwise ranking schemes have also been used in peer review settings to determine top candidates for recognition and accolades \cite{Sanderson23:3b1b}.
There are two main differences in our work from the aforementioned large volume of work on noisy sorting: (1) we introduce the verification oracle, which allows us to find the correct sorting, and our goal is to minimize the number of verifications used, and (2) we extend our results to a more general setting without a precise data-generation assumption and we characterize the number of verifications that are needed in these more general setting as well.

Noisy sorting has also been explored in the theoretical computer science literature \cite{feige1994computing}. There are many papers studying algorithmic aspects such as the optimal number of comparisons \cite{feige1994computing, Dereniowski21:Noisy, Wang22:Noisy, Gu23:Noisy}, parallel computation \cite{Braverman16:Parallel}, avoiding the use of resampling \cite{Braverman08:NoResampling, Geissmann19:NoResampling}, as well as estimation from mixture models \cite{LiuM18:Efficiently} and adversarial models \cite{LiuM23:Robust}.  Our techniques are also algorithmic, but our work initiates the study of recovering the true sorting perfectly while minimizing the number of oracle queries. At the same time, our results hold under much more flexible assumptions on the data-generating procedure.

\section{Model}\label{sec:model}
In our model, there is a set $S$ of $n$ elements, $x_{1}, \ldots, x_{n}$, each with a distinct cardinal value $\ell_i \in \mathbb{R}$.
We may think of $(x_{i}, \ell_{i})$ as key-value pairs.
If $\sigma: [n] \to [n]$ is the permutation that has $\ell_{\sigma(1)} < \ell_{\sigma(2)} < \ldots < \ell_{\sigma(n)}$, then we seek to order the elements as $(x_{\sigma(1)}, \ldots, x_{\sigma(n)}).$ We refer to this ordering as the ``ground-truth ordering.'' If $\sigma(i) = j$, we say that $x_{j}$ has order $i$.

To sort, we use two types of pairwise comparisons: noisy comparisons, which are widely available, and verifications, which are costly and intended to be used as little as possible. The algorithm's input is a set of noisy comparisons and the algorithm's performance is determined by: (1) its runtime, (2) its accuracy in reconstructing $\sigma$, and (3) the number of verifications it uses.

We represent the algorithm's input with a tournament graph $G = (V,E)$, allowing for bidirectional edges (also known as two-cycles). This graph may be constructed from a complete multi-graph by merging edges between every vertex pair while preserving their direction.
There are $n$ nodes $V = \{x_1, \dots, x_n\}$ and every edge represents a (potentially noisy) comparison. Specifically, a directed edge $(x_i, x_j)$ pointing from $x_i$ to $x_j$ indicates that $\ell_i > \ell_j$. We may have multiple comparisons per pair, and these comparisons may be conflicting or non-conflicting. Between any two nodes, we assume there is at most one edge in a given direction, so every node has an in- and out-degree of at most $n-1$. Moreover, there is at least one edge between every pair of nodes. If two nodes only have a unidirectional edge between them (i.e., there is no two-cycle between the nodes), we call that edge a \emph{simple} edge. We use the following terminology based on simple edges.

\begin{figure}[t]
	\begin{subfigure}[b]{0.48 \textwidth}
		\includegraphics[scale=1]{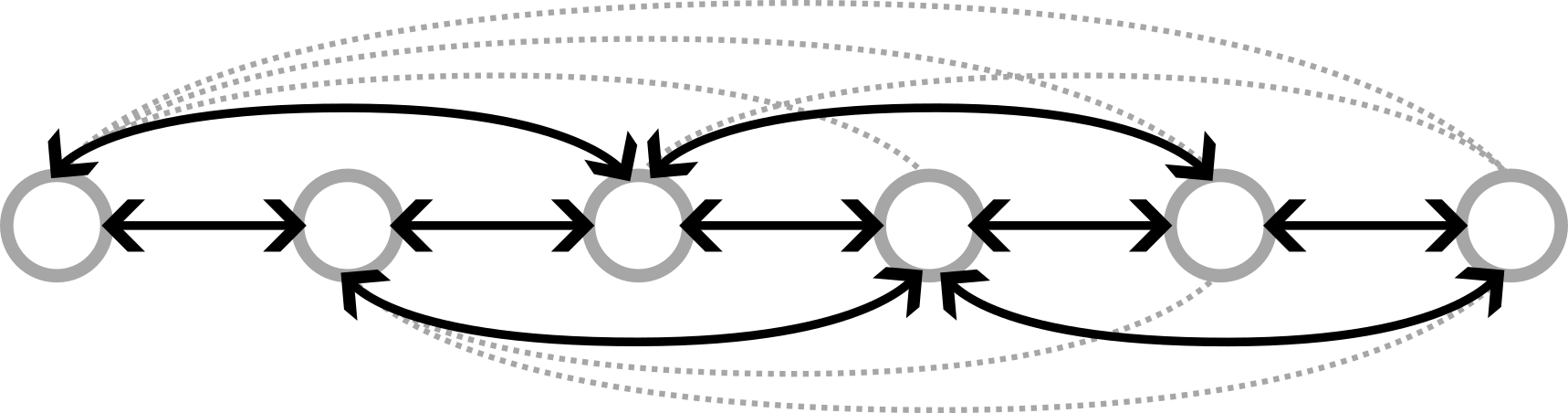}
		\centering
		\caption{Grey dotted lines are simple edges. Black solid lines are two-cycles.}
	\end{subfigure}
	\hfill
	\begin{subfigure}[b]{0.48\textwidth}
		\includegraphics[scale=1]{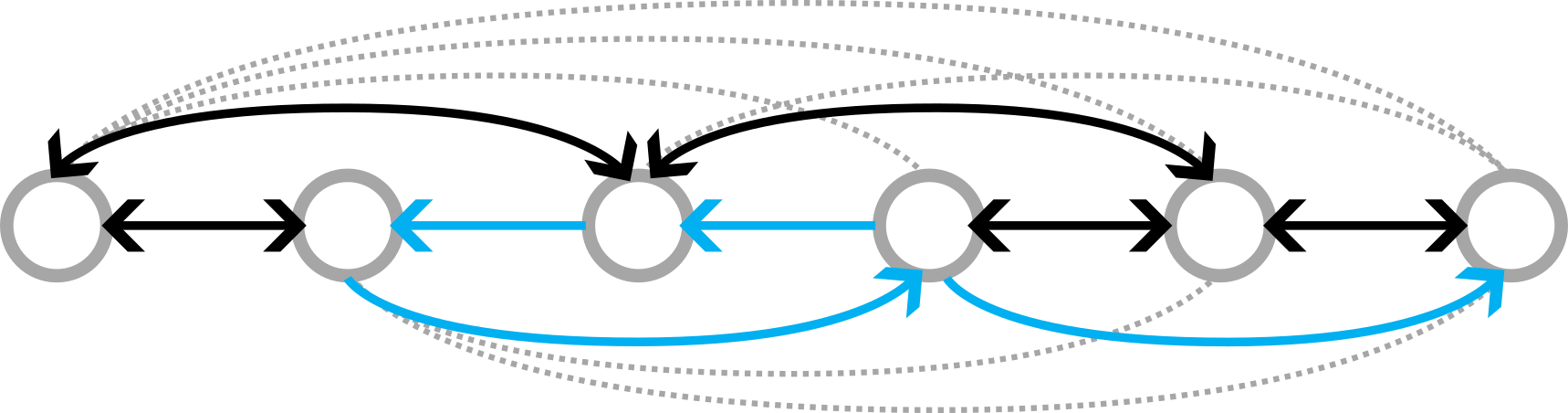}
		\centering
		\caption{Blue edges indicate ambiguous simple edges.\newline}
	\end{subfigure}
	\caption{Examples of edges in an input tournament graph.}
\end{figure}

\begin{definition}
The \emph{simple in-degree} of an element $x_i$, denoted $d_{in}(x_i)$, is the number of simple edges $(x_j, x_i)$. Similarly, the \emph{simple out-degree} $d_{out}(x_i)$ is the number of simple edges $(x_i, x_j)$.
\end{definition}

Our analysis builds off of a particularly simple type of corrupted tournament graph where a certain number of simple edges have been converted to two-cycles, defined as follows.
\begin{definition}
\label{def:ambiguous}
A tournament graph is \emph{$\vec{\nu}$-ambiguous} for $\vec{\nu} = (\nu_{+}, \nu_{-}) \in \mathbb{N}^{2}$ if:
    \begin{enumerate}
    \item All simple edges point in the correct direction, and
    \item Any element $x_i$ has exactly $\min\{\sigma(i)-1, \nu_{+}\}$ incident two-cycles with larger-indexed neighbors and exactly $\min\{n-\sigma(i), \nu_{-}\}$ incident two-cycles with smaller-indexed neighbors. That is, \begin{align}&|\{j \in V : \ell_j < \ell_i \land (x_j, x_i) \in E\}| = \min\{\sigma(i)-1, \nu_{+}\} \text{ and }\label{eq:in}\\
    &|\{j \in V : \ell_j > \ell_i \land (x_i, x_j) \in E\}| = \min\{n-\sigma(i), \nu_{-}\}\label{eq:out}.\end{align}
    \end{enumerate}
\end{definition}
This is an ``easy'' corrupted graph because, as we prove, we can uncover the ground-truth ordering without using any verifications. In Equation~\eqref{eq:in}, we take the minimum of $\sigma(i)-1$ and $\nu_+$ because there are at most $\sigma(i)-1$ elements $j \in V$ with $\ell_j < \ell_i$, and similarly for Equation~\eqref{eq:out}.

\paragraph{Adversarial model.} Our results apply to adversarially corrupted tournament graphs. If all edges were replaced with two-cycles, we would have no choice but to sort the elements using verifications. Therefore, we constrain the adversary by requiring them to start with any $\vec{\nu}$-ambiguous tournament graph. They may then change any number of two-cycles to simple edges (pointing in the correct or incorrect direction). The following type of edge will be useful for our analysis.
\begin{definition}
    A \emph{ambiguous simple} edge is a simple edge that the adversary switched from a two-cycle. More generally, a tournament graph has \emph{$k$ ambiguous simple edges} if it is constructed by taking a $\vec{\nu}$-ambiguous graph and changing exactly $k$ two-cycles to simple edges of any direction.
\end{definition}

We bound the number of verifications our algorithm requires using the number of ambiguous edges. Our algorithm does not need to know which edges are ambiguous.

\paragraph{Just noticeable difference model.} A natural instantiation of our model 
is the well-studied ``just-noticeable-difference'' (JND) model, which is parameterized by a value $\delta \geq 0$. In this model, we send pairwise-comparison queries to crowdsourcers. Given a query $(x_i,x_j)$, if $|\ell_i - \ell_j| 
 > \delta$, then a crowdsourcer will return the correct comparison with probability 1. Otherwise, a crowdsourcer is equally likely to return a correct or incorrect comparison. Thus, $\delta$ represents a threshold below which the difference between elements is not noticeable. We represent crowdsourcers' responses with edges in the tournament graph, so if $|\ell_i - \ell_j| < \delta$, there may be a two-cycle between $x_i$ and $x_j.$ 

\section{Results} \label{sec:results}
We begin by analyzing the fixed ambiguity model and the adversarial model, where all graphs are assumed to be $\vec{\nu}$-ambiguous, or constructed from a $\vec{\nu}$-ambiguous graph by changing two-cycles to simple edges.
In Section~\ref{sec:discrete_alg}, we describe the \textsf{CandidateSort} algorithm, prove correctness, and analyze runtime on the number of ambiguous simple edges of the input.
We proceed in Section~\ref{sec:jnd} to an average-case analysis of \textsf{CandidateSort} in the special-case setting of the JND model.
Then, we construct in Section~\ref{sec:lb} an input demonstrating optimality of \textsf{CandidateSort} on worst-case inputs.
\subsection{\textsf{CandidateSort} Algorithm}\label{sec:discrete_alg}

\begin{algorithm}[h]
    \caption{\textsf{CandidateSort}}\label{alg:candidatesort}
    \begin{algorithmic}
        \Require $V = \{x_{i} : i \in [n]\}, \vec{\nu} = (\nu_{+}, \nu_{-})$
        \State Define the set of the set of elements to be sorted: $U \gets V$
        \State Define the ascending process' candidate heap: $C_{A} \gets \emptyset$
        \State Define the descending process' candidate heap: $C_{D} \gets \emptyset$
        \State Define a null vector $\tau$ of length $n$ which will hold the final sorting
        \State Define the ascending (resp., descending) process' starting round: $r_{A} \gets 1$ (resp., $r_{D} \gets 1$)
        \State Define the \textbf{A}scending process as the starting process: $\rho \gets A$ \Comment{$\rho \in \{A, D\}$}
        \While{$U \neq \emptyset$}
            \State Collect the unsorted elements satisfying the degree threshold: \[S_{\rho} \gets \begin{cases}\left\{x \in U \setminus C_A : d_{in}(x) \geq \min\{n - r_A - \nu_{+}, 0\}\right\} &\text{if }\rho=A\\
            \left\{x \in U \setminus C_D : d_{out}(x) \geq \min\{r_D - 1 - \nu_{-}, 0\}\right\} &\text{if }\rho=D\end{cases}\]
            \If{$|S_{\rho}| + |C_{\rho}| = 1$} \Comment{Only one candidate element for this round}
                \State \textbf{do} \textsf{subroutine} \Comment{\textsf{subroutine} (Algorithm  \ref{alg:subroutine}) updates $C_\rho$, $\tau$, $U$, and $r_{\rho}$}
            \Else
                \State $\rho \gets \rho^{op}$ \Comment{Swap to ascending/descending process}
                \State Collect the unsorted elements satisfying the degree threshold: \[S_{\rho} \gets \begin{cases}\left\{x \in U \setminus C_A : d_{in}(x) \geq \min\{n - r_A - \nu_{+}, 0\}\right\} &\text{if }\rho=A\\
            \left\{x \in U \setminus C_D : d_{out}(x) \geq \min\{r_D - 1 - \nu_{-}, 0\}\right\} &\text{if }\rho=D\end{cases}\]
                \If{$|S_{\rho}| + |C_{\rho}| = 1$}
                    \State \textbf{do} \textsf{subroutine}
                \Else 
                    \State $\rho \gets \rho^{op}$
                    \State \textbf{do} \textsf{subroutine}
                    \State $\rho \gets \rho^{op}$
                \EndIf
            \EndIf
        \EndWhile
        \Return $\tau$
    \end{algorithmic}
\end{algorithm}
\begin{algorithm}[h] 
    \caption{\textsf{subroutine}}\label{alg:subroutine}
    \begin{algorithmic}
        \State Using verifications, insert the $S_{\rho}$ into the candidate heap $C_{\rho}$
        \State Define $\tau[r_{\rho}]_{\rho}$ to be an element popped from the heap $C_{\rho}$ 
        \State $U \gets U \setminus \{\tau[r_{\rho}]_{\rho}\}$ \Comment{Update unsorted elements}
        \State $r_{\rho} \gets r_{\rho} + 1$ \Comment{Update round}
    \end{algorithmic}
\end{algorithm}

The algorithm we propose, \textsf{CandidateSort} (Algorithm~\ref{alg:candidatesort}), is based on the following observations.
In the tournament graph for the true ordering, elements have unique simple in- and out-degrees of $n - i$ and $i - 1$, respectively, for the $i$-th order.
In any input graph, at most $\nu_{-} + \nu_{+}$ incident edges of any vertex are changed from the ground-truth tournament graph.
Hence, the simple in- and out-degrees for the $i$-th order are at least $\max\{n - i - \nu_{+}, 0\}$ and $\max\{i - 1 - \nu_{-}, 0\}$, respectively.

Our algorithm runs two simultaneous processes.
The first process, known as the \emph{ascending process}, determines elements in ascending order, and its steps are as follows.
We refer to the subroutine determining the element of order $k$ as \emph{round $k$}.
The ascending process starts at round $1$ and finishes at round $n$.
Across all rounds, it maintains a binary min-heap that we refer to as the \emph{candidate heap} $C_{A}$.
Comparisons in the heap are done with verifications.
On round $k$, we insert into the heap all elements with simple in-degree at least $\max\{n - k - \nu_{+}, 0\}$.
The minimum element is then popped from the heap, labeled as of order $k$, and discarded.
In the next round $k + 1$, we insert into the candidate heap all elements of simple in-degree $\max\{n - (k + 1) - \nu_{+}, 0\}$ not yet labeled or in the candidate heap, then pop the minimum and mark it as the element of order $k$.

The second process, referred to as the \emph{descending process}, involves essentially the exact same steps, except elements are determined in descending order.
Round $k$ determines the order of element $k$. 
The descending process starts at round $n$ and finishes at round $1$.
The min-heap is now a max-heap $C_{D}$, and at round $k$, we insert into the heap all elements of simple out-degree at least $\max\{k - 1 - \nu_{-}, 0\}$.
We pop the maximum element from this heap as the element of order $k$. 

The processes maintain separate heaps but share the same set of unsorted elements. 
Hence, if one process sorted and discarded element $x$, then the other discards $x$ by never adding it to its heap, or discarding $x$ from its heap if already inserted.
The two processes work in parallel as follows.
If there is a process whose first round requires no verifications, meaning it only adds one element to its heap, we start with that process.
Otherwise, we first start with round $1$ of the ascending process.
Usually, once a process has finished a round, the other process will proceed with its next round. 
As an example, we may start with round $1$ of the ascending process, continue to round $n$ of the descending process, then conduct round $2$ of the ascending process, etc.

However, there are caveats.
If the processes have altogether sorted all $n$ elements, colloquially called \emph{meeting in the middle}, then the algorithm terminates. This occurs when the ascending and descending processes finish rounds $k$ and $k + 1$, respectively, for some $k \in [n - 1]$.
Moreover, if a process had only one element in its heap during a given round, and hence used no verifications, then the same process will move to its next round, skipping a turn for the other process.
The process may continue having consecutive rounds until it adds multiple elements to its heap and therefore must use verifications. Then, it stops and allows the other process to proceed.

We now prove that \textsf{CandidateSort} returns the correct sorting. We will use the invariant that
$|C_{A}|, |C_{D}| \in \{0, 1, \ldots, \nu_{-} + \nu_{+}\}$ at any point in time.

\begin{theorem}[Correctness of \textsf{CandidateSort}]
\label{algo-correct}
\textsf{CandidateSort} sorts elements for graphs that are $\vec{\nu}$-ambiguous or constructed from $\vec{\nu}$-ambiguous graphs by changing two-cycles to simple edges.
\end{theorem}
\begin{proof}
We first prove correctness for a single ascending/descending process, before showing the ascending and descending processes are compatible in \textsf{CandidateSort}.

We argue the correctness of the ascending process and correctness of the descending process follows by symmetrical argument. We induct on the number of rounds. For $r_A = 1$, the simple in-degree lower bound ($d_{in}(x) \geq \min\{n - r_A - \nu_{+}, 0\}$) implies the element of order $1$ will be inserted into the candidate heap. Then, popping the minimum from the heap correctly determines the element of order $1$. If all elements of order $\leq k$ have been correctly determined in rounds $\leq k$, then in round $k + 1$, the element of order $k + 1$ has not been discarded. Moreover, the simple in-degree lower bound implies the element of order $k + 1$ will have been inserted into the heap by round $k + 1$. It follows that the element of order $k + 1$ must be the minimum in the heap at round $k + 1$, hence popping the minimum from the heap correctly determines the desired element.

We now prove compatibility of the ascending and descending processes. First, rounds for each process will never overlap. In other words, supposing $i < j$, if the ascending process determines the element of order $j$ (which must occur when $r_{A} = j$), then the descending process will never determine the element of order $i$ (which must occur when $r_{D} = n - i + 1$), and vice versa. This is because the processes determine elements from ``opposite ends'' and ``work inwards'' monotonically. Therefore, if the two processes ever ``cross paths,'' all elements will have been sorted.

The final concern is if the candidate heaps overlap at any round. This is a non-issue, since for a given process $P$, the other process $P'$ only removes elements that are not the one to be selected from the heap of $P$, because $P'$ currently seeks elements of different order, as argued previously.
\end{proof}

For our analysis, we relate the verification complexity to a quantity we now define.
\begin{definition}
An element $x_{i}$ of a graph $G$ makes $\app(x_{i})_{A} \in \mathbb{N}$ \emph{appearances} if it is in the ascending heap for $\app(x_{i})_{A}$ rounds, and $\app(x_{i})_{D}$ is defined likewise for descending. Let $\app(x_{i}) := \app(x_{i})_{A} + \app(x_{i})_{D}$ and let $\app(G) := \sum_{i = 1}^{n} \app(x_{i})$ be the number of total appearances for $G$.
\end{definition}

We now analyze the ``easiest'' case where the input is a $\vec{\nu}$-ambiguous tournament graph.

\begin{lemma}
\label{min-appearances-verifications}
Let $\nu_{-} + \nu_{+} < n$. 
In a $\vec{\nu}$-ambiguous tournament graph,
the number of total appearances is $n$ and the number of verifications is $0$.
\end{lemma}

\begin{proof}
Let $x_i$ be an element of order $k$. By Definition~\ref{def:ambiguous}, exactly $\min\{k-1, \nu_+\}$ vertices $x_{j} \in V$ with $\ell_j < \ell_i$ have a two-cycle between $x_{i}$ and $x_{j}$. The other $\max\{0, k-1-\nu_+\}$ satisfying $\ell_j < \ell_i$, have $(x_i, x_j) \in E$ is simple. Symmetrically, exactly $\min\{n-k, \nu_-\}$ vertices $x_{j} \in V$ with $\ell_j > \ell_i$ have two-cycles with $x_{i}$, and the other $\max\{0, n-k-\nu_-\}$ have a simple edge $(x_j, x_i) \in E.$ 
Therefore, $d_{in}(x) = \max\{n - \nu_{+} - k, 0\}$, implying only ever one vertex is in the ascending heap at round $i$ for $i \in [n - \nu_{+}]$.
Symmetrically, $d_{out}(x) = \max\{k - 1 - \nu_{-}, 0\}$, hence only one vertex is in the descending heap at round $i$ for $i \in \{\nu_{-} + 1, \nu_{-} + 2, \ldots, n\}$.
We start with the ascending process. The first $n - \nu_{+}$ elements in ascending order require no verifications since each round, we insert the sole element into the (empty) heap, then proceed to pop it as the minimum.
Therefore, these elements make one appearance in the heap. We now proceed by cases on the value of $\vec{\nu}$.

If $\nu_{+} > 1$, then for all elements $x$ with order $k > n-\nu_+$, we have that $d_{in}(x) = 0$. Therefore, at round $n - \nu_{+} + 1$, the ascending process must add $\nu_{+} > 1$ elements to the heap. Verifications are needed, so the descending process starts. By a symmetrical argument, the descending process uses 0 verifications for the remaining $\nu_{+} < n - \nu_{-}$ elements, which are its first $\nu_+$ elements.

If $\nu_+ = 1$, then round $n - \nu_{+} + 1 = n$ is the last round of the ascending process, so there is only one element remaining, which will be determined as the element of order $n$ with no verifications.

Thus, no verifications were used and all elements made one appearance, so that $\app(G) = n$.
\end{proof}

\begin{lemma}
\label{deg=>verif}
For an element $x$ of order $k_{1}$, let $d_{in}(x) = n - k_{1} + k_{2}$ and $d_{out}(x) = k_{1} - 1 + k_{3}$. Then, $\app(x) \leq 2  k_{2} + k_{3}$.
\end{lemma}
\begin{proof}
In the ascending process, $x$ enters the candidate heap when the simple in-degree lower bound $\min\{n - r_A - \nu_{+}, 0\}$ is $n - k_{1} + k_{2}$,
remaining in the heap until it is the minimum, which is when the threshold is at least $n - k_{1}$ by the induction in Theorem \ref{algo-correct}. Hence, $\app(x)_{A} \leq 1 + k_{2}$. Symmetrically, $\app(x)_{D} \leq 1 + k_{3}$. When the processes are working in tandem, $\app(x)$ can only decrease, since one processes will order $x$ before the other. Hence, $\app(x) \leq 2 + k_{2} + k_{3}$.
\end{proof}

The next result shows that the number of ambiguous simple edges implies an upper bound on $\app(G)$.

\begin{prop}
\label{monotonicity-appearances}
Let $G$ be a graph with at least $k$ two-cycles, and $H$ be constructed from $G$ by changing $k$ two-cycles to simple edges. Then, $\app(G) \leq \app(H) \leq \app(G) + 2k$, implying $\app(G)$ is non-decreasing as two-cycles in $G$ are changed to simple edges.
\end{prop}
\begin{proof}
Suppose $G'$ is a copy of $G$ except there is an edge $e = (u, v)$ which is simple in $G'$, but a two-cycle in $G$. It follows that $d_{in}(v)$ (resp. $d_{out}(u)$) will be 1 larger in $G'$ than $G$. Therefore, $\app(v)$ (resp. $\app(u)$) in the candidate heap of the ascending process (resp. descending process) in $G'$ is one more or the same as in $G$ by Lemma \ref{deg=>verif}. This means that $\app(G) \leq \app(G') \leq \app(G) + 2$. We repeat this process $k$ times to reach $H$ from $G$, implying $\app(G) \leq \app(H) \leq \app(G) + 2k$.
\end{proof}

We now give an upper bound the number of verifications for an input with exactly $k$ ambiguous simple edges.

\begin{theorem}[Upper bound for a fixed number of ambiguous simple edges]
\label{thm:general-result}
    Suppose $G$ is an input tournament graph with $k$ ambiguous simple edges. Then, the number of verifications required to obtain a correct sort with \textsf{CandidateSort} is $O(k)$.
\end{theorem}
\begin{proof}
    By Proposition \ref{monotonicity-appearances}, the $k$ ambiguous simple edges contribute at most $2k$ extra appearances.
    Since Lemma \ref{min-appearances-verifications} has $\app(H) = n$ for any $\vec{\nu}$-ambiguous graph $H$, and $G$ is built from some $H$ by changing $k$ two-cycles to simple edges, it follows from Proposition \ref{monotonicity-appearances} that $n \leq \app(G) \leq n + 2k$.
    Letting $c_{i}$ denote the number of elements in the candidate heap for the round determining element $i$, we therefore conclude that $O(\sum_{i = 1}^{2n} \log(c_{i})) \leq O(2n \log(\frac{2n + 2k}{2n})) = O(n \log(1 + \frac{k}{n})) \leq O(n \log(e^{k / n})) = O(n(k / n)) = O(k)$, where the first inequality follows from Lemma~\ref{uniform-upper-bound-verifications}.
\end{proof}

\subsection{JND Model}\label{sec:jnd}
We now instantiate our results under the JND model, where $\ell_1, \dots, \ell_n \in [n]$. In this case, $\vec{\nu} = (\delta, \delta)$.
An edge $(x_{i}, x_{j})$ where $|\ell_{i} - \ell_{j}| \leq \delta$ are sometimes referred to as \emph{$\delta$-close}, and when the edge in the graph is simple, it is a \emph{$\delta$-close simple} edge.
We study a setting where we have $r \in \mathbb{N}$ independent comparisons per pair $\{x_{i}, x_{j}\}$. Disagreeing comparisons give a two-cycle between $x_{i}$ and $x_{j}$ in the graph.
Otherwise, there is a simple edge.
Our analysis of \textsf{CandidateSort} implies an upper bound on expected verifications in the JND model. The full proof of the following theorem is in Appendix~\ref{app:proofs}.

\begin{restatable}{theorem}{discrete}
\label{expected-verifications-discrete}
The expected number of verifications to sort with \textsf{CandidateSort} is $O(n\delta 2^{-r})$. 
\end{restatable}
\begin{proof}[Proof sketch]
The probability that a $\delta$-close edge is simple is $2^{1 - r}$, and is independent of all other edges. Then, by Lemma \ref{discrete-num-delta-pairs}, the number of ambiguous simple edges is distributed as $X \sim \mathrm{Binom}(n\delta - \frac{\delta^{2} + \delta}{2}, p = 2^{1 - r})$. Hence, Lemma \ref{min-appearances-verifications} and Proposition \ref{monotonicity-appearances} imply that the number of total appearances in the ascending (and therefore descending) process is at most $2(X + n)$. By a similar argument as in the proof of Theorem~\ref{thm:general-result}, this implies that the expected verifications is
\[O\left(n\log\left(2^{1-r}\left(n\delta - \frac{\delta^2 + \delta}{2}\right) + n \right) - n\log\left(n\right) \right).\]
The result follows by simplifying this expression.
\end{proof}

\begin{corollary}
The following are instantiations of Theorem \ref{expected-verifications-discrete} for different ranges of $r$.
\begin{enumerate}
\item If $r \geq \log \delta$, then the expected number of verifications is $O(n)$.
\item If $r \geq \log n$, then the expected number of verifications is $O(\delta)$.
\item If $r \geq \log n + \log\delta$, then the expected number of verifications is $O(1)$.
\end{enumerate}
\end{corollary}

\subsection{Lower Bound}\label{sec:lb}
We conclude with a lower bound showing that Theorem~\ref{thm:general-result} is tight. Namely, we construct an input with $k$ ambiguous simple edges where any algorithm requires $\Omega(k)$ verifications to obtain the correct sorting.

\begin{theorem}\label{thm:lb}
    For $\nu_{+} = \nu_{-} < \lfloor \frac{n-2}{2} \rfloor$, $k = \lfloor \frac{n}{\nu_{-} + \nu_{+} + 2} \rfloor$, there are input tournament graphs with $2k$ ambiguous simple edges needing at least $k$ verifications.
\end{theorem}
\begin{proof}
First, let $\gamma = \nu_{+}$. take a $\nu$-ambiguous graph $G$ where $\ell_{i} = i$ for all $i \in [n]$, and a two-cycle is between $x_{i}$ and $x_{j}$ if and only if $|i - j| \leq \gamma$, referred to as being ``$\gamma$-close.''
Denote the element of order $i$ by $v_{i}$.
For each $i \in S := \{\gamma + 1, 3\gamma + 2, \ldots, (2k + 1)\gamma + k\}$, set edges as $v_{i} \leftrightarrow v_{i + 1}$, $v_{i - \gamma} \to v_{i}$, $v_{i + 1} \to v_{i + 1 + \gamma}$, where $\leftrightarrow$ indicates a two-cycle, and otherwise the edge has the specified direction.
All other $\gamma$-close edges for pairs $\{x_{j}, x_{k}\}$ for $i - \gamma \leq j < k \leq i + \gamma$ are two-cycles.
\begin{figure}[t]
\centering
		\includegraphics[scale=1.9]{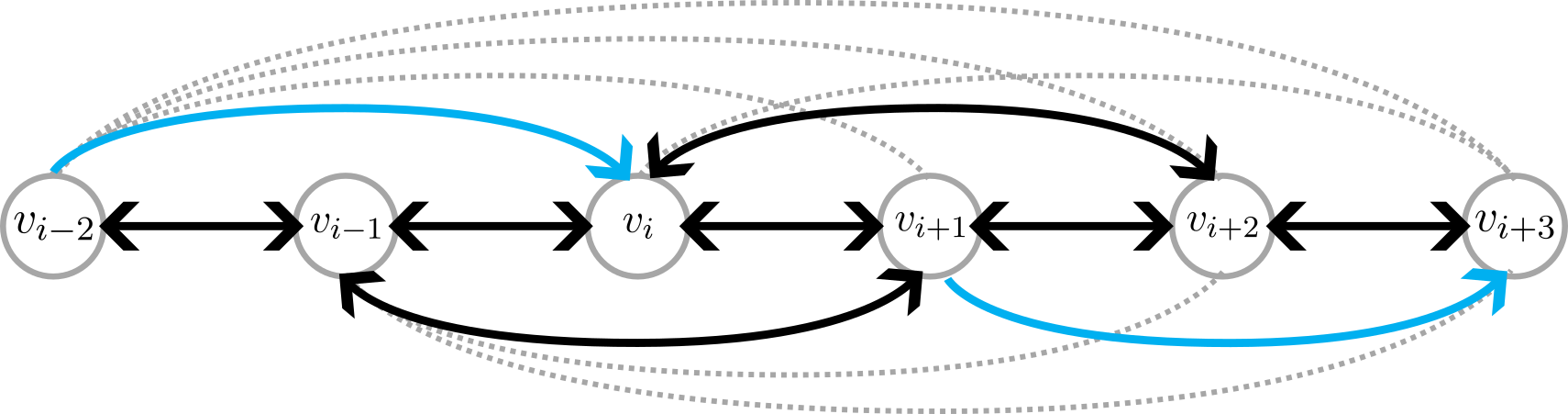}
	\caption{Illustration of the lower bound construction (Theorem~\ref{thm:lb}) with $\gamma = 2$. Grey dotted lines are simple edges, black solid lines are two-cycles, and blue solid lines are ambiguous simple edges. Each construction consists of $k$ ``copies'' of this graph and the necessary non-ambiguous edges.}\label{fig:lb}
\end{figure}
We illustrate this graph in Figure~\ref{fig:lb}.
The edges originally labeled with $\leftrightarrow$ are \emph{middle-edges}, whereas the others are \emph{side-edges}.
We will argue each pair $(v_{i}, v_{i + 1})$ for $i \in S$ is indistinguishable in the input, meaning there is another input graph appearing identical, but $v_{i}$ and $v_{i + 1}$ are swapped in the order.

For each $(2\gamma + 2)$-vertex subgraph $G_{i} = (V_{i} = \{2(i-1)\gamma + i, 2(i-1)\gamma + i + 1, \ldots, 2i\gamma + i\}, E_{i})$ where $i \in [k]$, there is a graph isomorphism given by $v_{i} \mapsto v_{i + 1}$ and $v_{i + 1} \mapsto v_{i},$ with identity elsewhere, since any $v \in V_{i} \setminus \{v_{i}, v_{i + 1}\}$ has edge orientations $\{v, v_{i}\}$ and $\{v, v_{i + 1}\}$ identical. Moreover, $v_{i} \leftrightarrow v_{i + 1}$ is a two-cycle.  
Note that vertices outside $G_{i}$ cannot distinguish $v_{i}$ and $v_{i + 1}$ since they are not $\gamma$-close to both.
Since all inputs are complete graphs, it is the edge-related data determining a graph up to isomorphism.
This makes $G_{i}$ indistinguishable from its isomorphic counterpart without verifying the edge $(v_{i}, v_{i + 1})$. Hence, each subgraph requires at least one verification and the input needs at least $k$ verifications to be distinguishable.
\end{proof}

\section{Conclusion}
We studied a novel noisy sorting model that generalizes the JND model, with motivations from crowdsourcing applications.
This setting seeks to better capture the ambiguity in such scenarios, moving away from the more uniform assumptions of well-studied noise models.
Our algorithm demonstrates how a sparing use of expert feedback can achieve strong results, such as deterministic sorting, even in such a difficult environment.

\bibliographystyle{plainnat}
\bibliography{refs}

\appendix
\section{Appendix}\label{app:proofs}

We make a combinatorial observation about the tournament graph for the JND model.

\begin{lemma}
\label{discrete-num-delta-pairs}
The number of $\delta$-close edges is ${\delta \choose 2} + (n - \delta)\delta = n\delta - \frac{\delta^{2} + \delta}{2}$.
\end{lemma}
\begin{proof}
Order the elements by their ground-truth ordering, in ascending order. Among the first $\delta$ elements, there are evidently ${\delta \choose 2}$ $\delta$-close edges. Then, every consecutive element $x$ afterward contributes $\delta$ additional $\delta$-close edges with the greatest $\delta$ elements with cardinal value less than $x$. Since there are $n - \delta$ such elements, we have that the total number of $\delta$-close edges is ${\delta \choose 2} + (n - \delta)\delta$.
\end{proof}

We use the following lemma to bound the expected value of $\app(G)$ in the JND model. 

\begin{lemma}[Log-Binomial Expectation]
\label{log-binomial-expectation}
If $X \sim \mathrm{Binom}(n, p)$,
then $\mathbb{E}[\log\left(\frac{X + n}{n}\right)] \leq \log(np + n) - \log(n) - \frac{p(1 - p)}{7n} + O\left(\frac{1}{n^{2}}\right)$.
\end{lemma}
\begin{proof}
By expanding and bounding the Taylor expansion, we make the following derivation.
\begin{flalign*}
\mathbb{E}[\log(X + n)] & \leq \log(np + n) + \frac{\mathbb{E}[X - np]}{np + n} - \frac{\mathrm{Var}(X)}{(np + n)^{2}} + O\left(\frac{1}{n^{2}}\right) \\
& \leq \log(np + n) + \frac{\mathbb{E}[X - np]}{np + n} - \frac{np(1-p)}{7n^{2}} + O\left(\frac{1}{n^{2}}\right) \\
& \leq \log(np + n) + \frac{\mathbb{E}[X - np]}{np + n} - \frac{p(1-p)}{7n} + O\left(\frac{1}{n^{2}}\right).
\end{flalign*}
\end{proof}

The proceeding observation is used to provide an upper bound on the number of verifications in Theorem \ref{thm:general-result}.

\begin{lemma}
\label{uniform-upper-bound-verifications}
Suppose $M \in \mathbb{N}_{\geq n}$ and $c_{1}, \ldots, c_{n} \in \mathbb{N}$ has $\sum\limits_{i = 1}^{n} c_{i} = M$. Then, $\sum\limits_{i = 1}^{n} \log(c_{i}) \leq n \log\left(\frac{M}{n}\right)$.
\end{lemma}
In other words, Lemma~\ref{uniform-upper-bound-verifications} implies that $\sum\log(c_{i})$ is maximized when the ``mass'' of $M$ is distributed uniformly across $c_{i}$.

\begin{proof}
We relax the problem to ignore integer constraints. Our program is
\begin{equation*}
    \begin{aligned}
        \max_{c} \quad   & \sum\limits_{i = 1}^{n} \log(c_{i}) \\
        \text{s.t. } \quad  
        & \sum_{i = 1}^{n} c_{i} = M, \\
        & c \geq 0. \\
    \end{aligned}
\end{equation*}
We show that the optimal solution to this maximization problem is $c^{*}_{i} = \frac{M}{n}$ for all $i \in [n]$. The program is convex and satisfies Slater's condition, hence strong duality holds for the Lagrangian dual problem, where the Lagrangian is given by $\mathcal{L}(c, \nu) = \sum\limits_{i = 1}^{n} \log(c_{i}) + \nu(M - \sum\limits_{i = 1}^{n} c_{i})$. Then, $\nabla_{c} \mathcal{L}(c, \nu) = 0 \implies \frac{1}{c_{i}} - \nu = 0 \implies c^{*}_{i} = \frac{1}{\nu}$ for all $i \in [n]$, and we have that the dual problem is simply to maximize $n\log\left(\frac{1}{\nu}\right) + \nu(M - \frac{n}{\nu})$. It follows that $\nu^{*} = \frac{n}{M}$ and hence $c^{*}_{i} = \frac{M}{n}$ for all $i \in [n]$.

Since all solutions satisfying the integer constraint are in the feasible set of the relaxed program, it is clear that the optimal solution for the relaxed program is an upper bound for all solutions of the relaxed program.
\end{proof}

\discrete*

\begin{proof}
The probability that a $\delta$-close edge is simple is $2^{1 - r}$, and is independent of all other edges. Then, by Lemma \ref{discrete-num-delta-pairs}, the number of ambiguous simple edges is distributed as $X \sim \mathrm{Binom}(n\delta - \frac{\delta^{2} + \delta}{2}, p = 2^{1 - r})$. Hence, Lemma \ref{min-appearances-verifications} and Proposition \ref{monotonicity-appearances} imply that the number of total appearances in the ascending (and therefore descending) process is at most $2(X + n)$. Lemma \ref{log-binomial-expectation} implies that 
\[\mathbb{E}\left[n\log\left(\frac{2}{n}(X + n)\right)\right] = O\left( n \log \left(  2^{1 - r}\left(n\delta - \frac{\delta^{2} + \delta}{2}\right) + n \right) - n\log \frac{n}{2} \right).\]
Therefore, by Lemma \ref{uniform-upper-bound-verifications}, the number of expected verifications is at most
\begin{flalign*}
& O\left(n\log\left(2^{1-r}\left(n\delta - \frac{\delta^2 + \delta}{2}\right) + n \right) - n\log\left(n\right) \right) \\
= & \quad O\left(n\log(n) + n\log\left(1 + 2^{1-r}\left(\delta - \frac{\delta^{2}}{2n}\right)\right) - n\log(n)\right) \tag*{using $a + b = a(1 + b/a)$} \\
= & \quad O\left( n\log\left(1 + 2^{1-r}\delta\right)\right) \\
= & \quad O\left( n \log\exp(2^{1-r}\delta) \right) \tag*{by Taylor expansion of $\exp$} \\
= & \quad O\left(n\delta 2^{-r}\right).
\end{flalign*}
\end{proof}

\end{document}